%% file: MAIN.tex
\documentclass{IEEEtran}
\IEEEoverridecommandlockouts

\makeatletter
\def\endthebibliography{%
  \def\@noitemerr{\@latex@warning{Empty `thebibliography' environment}}%
  \endlist
}
\makeatother

\usepackage[left=1.2cm, right=1.2cm, top=1.8cm, bottom=3.7cm]{geometry}
\usepackage[cmex10]{amsmath}
\usepackage{amssymb,amsthm}

\usepackage{bm}
\usepackage{amsfonts}
\usepackage{relsize}
\usepackage{cleveref}
\usepackage{cuted}
\usepackage{subfig}
\usepackage{url}
\usepackage{graphicx}
\usepackage{epstopdf}
\usepackage{algorithm}
\usepackage[noend]{algpseudocode}
\usepackage{bbm}
\usepackage{comment}
\usepackage{multirow}
\usepackage{booktabs}
\usepackage{textcomp}
\usepackage{cite}
\usepackage{subfig}
\usepackage{xcolor}
\def\BibTeX{{\rm B\kern-.05em{\sc i\kern-.025em b}\kern-.08em
    T\kern-.1667em\lower.7ex\hbox{E}\kern-.125emX}}
\usepackage{enumitem} 
\newtheorem{lemma}{Lemma}

\makeatletter

\algnewcommand{\Inputs}[1]{%
  \State \textbf{Inputs:}
  \Statex \hspace*{\algorithmicindent}\parbox[t]{.8\linewidth}{\raggedright #1}
}
\algnewcommand{\Initialize}[1]{%
  \State \textbf{Initialize:}
  \Statex \hspace*{\algorithmicindent}\parbox[t]{.8\linewidth}{\raggedright #1}
}

\def\BState{\State\hskip-\ALG@thistlm}
\makeatother

\title{Capacity-Approaching Autoencoders for Communications}

\author{\IEEEauthorblockN{Nunzio A. Letizia, \textit{Student Member, IEEE}\thanks{The authors are with the University of Klagenfurt - Chair of Embedded Communication Systems, 9020 Klagenfurt, Austria. (e-mail: \{nunzio.letizia, andrea.tonello\}@aau.at)} and Andrea M. Tonello, \textit{Senior Member, IEEE}} 
}

\begin{document}
\maketitle
\thispagestyle{plain}

\begin{abstract}
The autoencoder concept has fostered the reinterpretation and the design of modern communication systems. It consists of an encoder, a channel, and a decoder block which modify their internal neural structure in an end-to-end learning fashion. However, the current approach to train an autoencoder relies on the use of the cross-entropy loss function. This approach can be prone to overfitting issues and often fails to learn an optimal system and signal representation (code). In addition, less is known about the autoencoder ability to design channel capacity-approaching codes, i.e., codes that maximize the input-output information under a certain power constraint. The task being even more formidable for an unknown channel for which the capacity is unknown and therefore it has to be learnt. 

In this paper, we address the challenge of designing capacity-approaching codes by incorporating the presence of the communication channel into a novel loss function for the autoencoder training. In particular, we exploit the mutual information between the transmitted and received signals as a regularization term in the cross-entropy loss function, with the aim of controlling the amount of information stored. By jointly maximizing the mutual information and minimizing the cross-entropy, we propose a methodology that a) computes an estimate of the channel capacity and b) constructs an optimal coded signal approaching it. 
Several simulation results offer evidence of the potentiality of the proposed method.
\end{abstract}


\maketitle
\input{Introduction.tex}


\input{Theory.tex}

\input{Results.tex}

\input{Conclusions.tex}

\bibliographystyle{IEEEtran}
\bibliography{IEEEabrv,biblio}

\input{Appendix.tex}


\end{document}

%% file: Introduction.tex
\section{Introduction}
\label{sec:introduction}
Communication systems have reached a high degree of performance, meeting demanding requirements in numerous application fields due to the ability to cope with real-world effects exploiting various accomplished physical system models. Reliable transmission in a communication medium has been investigated in the milestone work of C. Shannon \cite{Shannon1948} who suggested to represent the communication system as a chain of fundamental blocks, i.d., the transmitter, the channel and the receiver. Each block is mathematically modelled in a bottom-up fashion, so that the full system results mathematically tractable. 

On the contrary, machine learning (ML) algorithms take advantage of the ability to work with and develop top-down models. In particular, deep learning (DL) has recently experienced a strong growth thanks to a larger availability of labeled data and increased computational power of processing units. Several fields have borrowed techniques from ML resulting in significant contributions and research progress. However, only in recent years, researchers have started adopting ML tools in the communication domain with promising results \cite{OsheaIntro,ML_PLC,Dorner2018}. In particular, the communication chain has been reinterpreted as an autoencoder-based system \cite{OsheaIntro}, a deep neural network (NN) which takes as input a sequence of bits $s$, produces a coded signal, feeds it into a channel layer and tries to reconstruct the initial sequence from the channel output samples. The intermediate channel layer implicitly separates a prior neural block, the encoder, from the posterior one, the decoder. The encoder maps bits into symbols to be transmitted, ideally placing them into a coded signal that changes during the training process in order to mitigate the effects of the channel and noise. The decoder, instead, performs a classification task and predicts the input message $\hat{s}$ from the received signal samples.
The autoencoder can be trained end-to-end such that the block-error rate (BLER) of the full system is minimized. 
This idea pioneered a number of related works aimed at showing the potentiality of deep learning methods applied to wireless communications \cite{Dorner2018,TurboAE,Hoydis2019,Wunder2019}. In \cite{Dorner2018}, communication over-the-air has been proved possible without the need of any conventional signal processing block, achieving competitive bit error rates w.r.t. to classical approaches. Turbo autoencoders \cite{TurboAE} reached state-of-the-art performance with capacity-approaching codes at a low signal to noise ratio (SNR). These methods rely on the a-priori channel knowledge (most of the time a Gaussian noise intermediate layer is assumed) and they fail to scale when the channel is unknown. To model the intermediate layers representing the channel, one approach is to use generative adversarial networks (GANs) \cite{Goodfellow2014}. GANs are a pair of networks in competition with each other: a generator model $G$ that captures the data distribution and a discriminator model $D$ that distinguishes if a sample is a true sample coming from real channel samples rather than a fake one coming from samples generated by $G$. In this way, the generator implicitly learns the channel distribution $p_Y(y|x)$, resulting in a differentiable network which can be jointly trained in the autoencoder model \cite{OsheaGAN,RighiniLetizia2019,Letizia2019a}. 

None of the aforementioned methods explicitly considered the information rate in the cost function. In this direction, the work in \cite{Hoydis2019} included the information rate in the formulation and leveraged autoencoders to jointly perform geometric and probabilistic signal constellation shaping. If the channel model is not available, the encoder can be independently trained to learn and maximize the mutual information between the input and output channel samples, as presented in \cite{Wunder2019}. However, the decoder is independently designed from the encoder and it does not necessarily grant error-free decoding. Therefore, the encoder and decoder learning process shall be done \textit{jointly}. In addition, the cost function used to train the autoencoder shall be appropriately chosen. With this goal in mind, let us firstly look into the historical developments and progresses made in the ML field.

The autoencoder was firstly introduced as a non-linear principle component analysis method, exploiting neural networks \cite{Kramer1991}. Indeed, the original network contained an internal bottleneck layer which forced the autoencoder to develop a compact and efficient representation of the input data, in an unsupervised manner.
Several extensions have been further investigated, such as the denoising autoencoder (DAE) \cite{Vincent2008}, trained to reconstruct corrupted input data, the contractive autoencoder (CAE) \cite{Rifai2011}, which attempts to find a simple encoding and decoding function, and the $k$-sparse autoencoder \cite{Makhzani2013} in which only $k$ intermediate nodes are kept active.
Autoencoders find application also in generative models as described in \cite{Kingma2013}. However, all of them are particular forms of \textit{regularized} autoencoders. Regularization is often introduced as a penalty term in the cost function and it discourages complex and extremely detailed models that would poorly generalize on unseen data. In this context, the information bottleneck Lagrangian \cite{Tishby1999} was used as a regularization term to study the sufficiency (fidelity) and minimality (complexity) of the internal representation \cite{Alemi2017,Soatto2018}. So far, in the context of communication systems design, regularization in the loss function has not been introduced yet. In addition, the decoding task is usually performed as a classification task. In ML applications, classification is usually carried out by exploiting a final softmax layer together with the categorical cross-entropy loss function. The softmax layer provides a probabilistic interpretation of each possible bits sequence so that the cross-entropy measures the dissimilarity between the reference and the predicted sequence of bits distribution, $p(s)$ and $q(\hat{s})$, respectively.
Nevertheless, training a classifier via cross-entropy suffers from the following problems: firstly, it does not guarantee any optimal latent representation. Secondly, it is prone to overfitting issues, especially in the case of large networks. 
Lastly, in the particular case of autoencoders for communications, the fundamental trade-off between the rate of transmission and reliability, namely, the channel capacity $C$, is not explicitly considered in the learning phase. 

These observations made us rethinking the problem by formulating the two following questions:
\begin{itemize}
\item[a)] Given a power constraint, is it possible to design capacity-approaching codes exploiting the principle of autoencoders?
\item[b)] Given a power constraint, is it possible to estimate channel capacity with the use of an autoencoder?
\end{itemize}
The two questions are inter-related and the answer of the first one provides an answer to the second one in a constructive way, since if such a code is obtained, then the distribution of the input signal that maximizes the mutual information is also determined, and consequentially the channel capacity can also be obtained.

Inspired by the information bottleneck method \cite{Tishby1999} and by the notion of channel capacity, a novel loss function for autoencoders in communications is proposed in this paper. The amount of information stored in the latent representation is controlled by a regularization term estimated using the recently introduced mutual information neural estimator (MINE) \cite{Mine2018}, enabling the design of nearly optimal codes. To the best of our knowledge, it is the first time that the influence of the channel appears in the end-to-end learning phase in terms of mutual information. More specifically, the contributions of the paper are the following:
\begin{itemize}
\item A new loss function is proposed. It enables a new signal constellation shaping method.
\item Channel coding is obtained by \textit{jointly} minimizing the cross-entropy between the input and decoded message, and maximizing the mutual information between the transmitted and received signals. 
\item A regularization term $\beta$ controls the amount of information stored in the symbols for a fixed message alphabet dimension $M$ and a fixed rate $R<C$, playing as a trade-off parameter between error-free decoding ability and maximal information transfer via coding. The NN architecture is referred to as \textit{rate-approaching autoencoder}.
\item By including the mutual information, we propose a new theoretical iterative scheme to built capacity-approaching codes of length $n$ and rate $R$ and consequently estimate channel capacity. This yields a scheme referred to as \textit{capacity-approaching autoencoder}.
\item With the notion of explainable ML in mind, the rationale for the proposed metric and methodology is discussed in more fundamental terms following the information bottleneck method \cite{Tishby1999} and by discussing the cross-entropy decomposition.  
\end{itemize} 

The remainder of the paper is organized as follows: In Sec.\ref{sec:theory}, we briefly review the autoencoder principle and starting from the channel capacity concept, we intuitively motivate the presence of the mutual information in a new loss function. Sec.\ref{sec:IB} discusses the mathematical foundation behind the new regularization term. In Sec.\ref{sec:capacity}, we exploit the mutual information block previously introduced to design capacity-approaching codes and to learn the channel capacity. Sec.\ref{sec:results} validates the proposed methodology through numerical results. Finally, Sec.\ref{sec:conclusions} reports the conclusions.


%% file: Theory.tex
\subsection*{\textbf{Notation}}
\label{subsec:notation}
$X$ denotes a multivariate random variable of dimension $d$, while $x\in \mathcal{X}$ denotes its realization. Vectors $\mathbf{y}$ and matrices  $\mathbf{Y}$ are represented using lower case bold and upper case bold letters, respectively. $p_Y(y|x)$ and $p_{XY}(x,y)$ represent the conditional and joint probability density functions, while $p_X(x)p_Y(y)$ is the product of the two marginals. $I(X;Y)$ denotes the mutual information between the random variables $X$ and $Y$.

\section{Rate-Approaching Autoencoder}
\label{sec:theory}
In this section, we introduce an autoencoder architecture to design a coding scheme that reaches a certain rate under a certain power constraint and code length. Then, we present the motivations behind the need of a new design metric that accounts for the mutual information in the classical cross-entropy loss function. 

\subsection{End-to-end Autoencoder-based Communications}
\label{subsec:end-to-end}
The communication chain can be divided into three fundamental blocks: the transmitter, the channel, and the receiver. The transmitter attempts to communicate a message $s\in \mathcal{M} = \{1,2,\dots,M\}$. To do so, it transmits $n$ complex baseband symbols $\mathbf{x}\in \mathbb{C}^{n}$ at a rate $R=(\log_2 M)/n$ (bits per channel use) over the channel, under a power constraint. In general, the channel modifies $\mathbf{x}$ into a distorted and noisy version $\mathbf{y}$. The receiver takes as input $\mathbf{y}$ and produces an estimate $\hat{s}$ of the original message $s$. From an analytic point of view, the transmitter applies a transformation $f:\mathcal{M} \to \mathbb{C}^{n}$, $\mathbf{x}=f(s)$ where $f$ is referred to as the \textit{encoder}. The channel is described in probabilistic terms by the conditional transition probability density function $p_Y(y|x)$. The receiver, instead, applies an inverse transformation $g: \mathbb{C}^{n} \to \mathcal{M}$, $\hat{s}= g(\mathbf{y})$ where $g$ is referred to as the \textit{decoder}. Such  communication scheme can be interpreted as an autoencoder which learns internal robust representations $\mathbf{x}$ of the messages $s$ in order to reconstruct $s$ from the perturbed channel samples $\mathbf{y}$ \cite{OsheaIntro}. 

The autoencoder is a deep NN trained end-to-end using stochastic gradient descent (SGD). The encoder block $f(s;\theta_E)$ maps $s$ into $\mathbf{x}$ and consists of an embedding layer followed by a feedforward NN and a normalization layer to fulfill a given power constraint. The channel is identified with a set of layers; a canonical example is the AWGN channel, a Gaussian noise layer which generates $y_i = x_i+w_i$ with $w_i\sim \mathcal{CN}(0,\sigma^2), i=1,\dots,n$. The decoder block $g(\mathbf{y};\theta_D)$ maps the received channel samples $\mathbf{y}$ into the estimate $\hat{s}$ by building the empirical probability density function $p_{\hat{S}|Y}(\hat{s}|y;\theta_D)$. It consists of a feedforward NN followed by a softmax layer which outputs a probability vector of dimension $M$ that assigns a probability to each of the possible $M$ messages. The encoder and decoder parameters $(\theta_E, \theta_D)$ are jointly optimized during the training process with the objective to minimize the cross-entropy loss function
\begin{equation}
\label{eq:CE}
\mathcal{L}(\theta_E, \theta_D) = \mathbb{E}_{(x,y)\sim p_{XY}(x,y)}[-\log(g(\mathbf{y};\theta_D))],
\end{equation}

whereas the performance of the autoencoder-based system is typically measured in terms of bit error rate (BER) or block error rate (BLER)
\begin{equation}
\label{eq:BLER}
P_e =  {P[\hat{s}\neq s]}.
\end{equation}

The autoencoder ability to learn joint coding and modulation schemes \cite{OsheaIntro,Dorner2018} for any type of channel (even for those without a known model) and for any type of non-linear effects (e.g. amplifiers and clipping) \cite{OFDM_AE} demonstrates the potentiality and flexibility of the approach.

However, the cross-entropy loss function does not guarantee any optimality in the code design and it is often prone to overfitting issues \cite{Tishby2017, Soatto2018}. In addition and most importantly, optimal system performance is measured in terms of achievable rates, thus, in terms of mutual information $I(X;Y)$ between the transmitted $\mathbf{x}$ and the received signals $\mathbf{y}$, defined as
\begin{equation}
\label{eq:mutual_information}
I(X;Y) = \mathbb{E}_{(x,y)\sim p_{XY}(x,y)}\biggl[\log\frac{p_{XY}(x,y)}{p_X(x)p_Y(y)}\biggr].
\end{equation}

In communications, the trade-off between the rate of transmission and reliability is expressed in terms of channel capacity. For a memory-less channel, the capacity is defined as
\begin{equation}
\label{eq:capacity}
C = \max_{p_X(x)} I(X;Y),
\end{equation}
where $p_X(x)$ is the input signal probability density function. Finding the channel capacity $C$ is at least as complicated as evaluating the mutual information. As a direct consequence, building capacity-approaching codes results in a formidable task. 

Given a certain power constraint and rate $R$, the autoencoder-based system, that is trained to minimize the cross-entropy loss function, is able to automatically build zero-error codes. Nevertheless, there may exist a higher rate error-free code. Therefore, the autoencoder does not provide a capacity-achieving code. In other words, conventional autoencoding approaches, through cross-entropy minimization, allow to obtain excellent decoding schemes. Nevertheless, no guarantee to find an optimal encoding scheme is given, especially in deep NNs where problems such as vanishing and exploding gradients occur \cite{279181}. 
Hence, the starting point to design capacity-approaching codes is to redefine the loss function used by the autoencoder.
In detail, we propose to include the mutual information quantity as a regularization term. The proposed loss function reads as follows

\begin{equation}
\label{eq:CE_MI}
\hat{\mathcal{L}}(\theta_E, \theta_D) = \mathbb{E}_{(x,y)\sim p_{XY}(x,y)}[-\log(g(y;\theta_D))]-\beta I(X;Y).
\end{equation}

The loss function in \eqref{eq:CE_MI} forces the autoencoder to jointly modify the network parameters $(\theta_E,\theta_D)$. The decoder reconstructs the original message $s$ with lowest possible error probability $P_e$, while the encoder finds the optimal input signal distribution $p_X(x)$ which maximizes $I(X;Y)$, for a given rate $R$ and code length $n$ and for a certain power constraint. We will denote such type of trained autoencoder as \textit{rate-approaching} autoencoder.
It should be noted that such a NN architecture does not necessarily provide an optimal code capacity-wise, since we set a target rate which does not correspond to channel capacity. To solve this second objective, in Sec.\ref{sec:capacity} we will describe a methodology leading to a new scheme that we name \textit{capacity-approaching} autoencoder.

To compute the mutual information $I(X;Y)$, we can exploit recent results such as MINE \cite{Mine2018}, as discussed below.

\subsection{Mutual Information Estimation}
\label{subsec:mutual_information_estimation}
The mutual information between two random variables, $X$ and $Y$, is a fundamental quantity in statistics and information theory. It measures the amount of information obtained about $X$ by observing $Y$. The difficulty in computing $I(X;Y)$ resides in its dependence on the joint probability density function $p_{XY}(x,y)$, which is usually unknown. Common approaches to estimate the mutual information rely on binning, density and kernel estimation \cite{Moon1995}, $k$-nearest neighbours \cite{Kraskov2004}, $f$-divergence functionals \cite{Nguyen2010}, and variational lower bounds. 

Recently, the MINE estimator \cite{Mine2018} proposed a NN-based method to estimate the mutual information. MINE is based on the Donsker-Varadhan dual representation \cite{Donsker1983} of the Kullback-Leibler divergence, in particular 
\begin{equation}
D_{\text{KL}}(p||q) = \sup_{T:\Omega \to \mathbb{R}} \mathbb{E}_{x\sim p(x)}[T(x)]-\log(\mathbb{E}_{x\sim q(x)}[e^{T(x)}]),
\end{equation}
where the supremum is taken over all functions $T$ such that the expectations are finite. Indeed, by parametrizing a family of functions $T_{\theta} : \mathcal{X}\times \mathcal{Y} \to \mathbb{R}$ with a deep NN with parameters $\theta \in \Theta$, the following bound \cite{Mine2018} holds
\begin{equation}
\label{eq:lower_bound}
I(X;Y)\geq I_\theta(X;Y),
\end{equation}
where $I_{\theta}(X;Y)$ is the neural information measure \cite{Mine2018} defined as
\begin{align}
\label{eq:MINE}
I_\theta(X;Y) = \sup_{\theta \in \Theta}& \; \mathbb{E}_{(x,y)\sim p_{XY}(x,y)}[T_{\theta}(x,y)] \nonumber \\
& -\log(\mathbb{E}_{(x,y)\sim p_X(x) p_Y(y)}[e^{T_{\theta}(x,y)}]).
\end{align}
The neural information $I_\theta(X;Y)$ in \eqref{eq:MINE} can be maximized using back-propagation and gradient ascent, leading to a tighter bound in \eqref{eq:lower_bound}. The consistency property of MINE guarantees the convergence of the estimator to the true mutual information value.

Estimating the mutual information $I(X;Y)$ is not enough to build capacity-approaching codes for a generic channel. A maximization over all possible input distribution $p_X(x)$ is also required. Therefore, to learn an optimal scheme, at each iteration the encoder needs both the cross-entropy gradient for the decoding phase and the mutual information gradient, from MINE, for the optimal input signal distribution. 
The proposed loss function in \eqref{eq:CE_MI} (see also Fig.\ref{fig:CE_MI}) shows such double role. In this way, the autoencoder trained with the new loss function intrinsically designs codes for which the mutual information $I(X;Y)$ is known and maximal by construction, under the aforementioned constraints of power, rate $R$ and code-length $n$.

\begin{figure}
	\centering
	\includegraphics[scale=0.42]{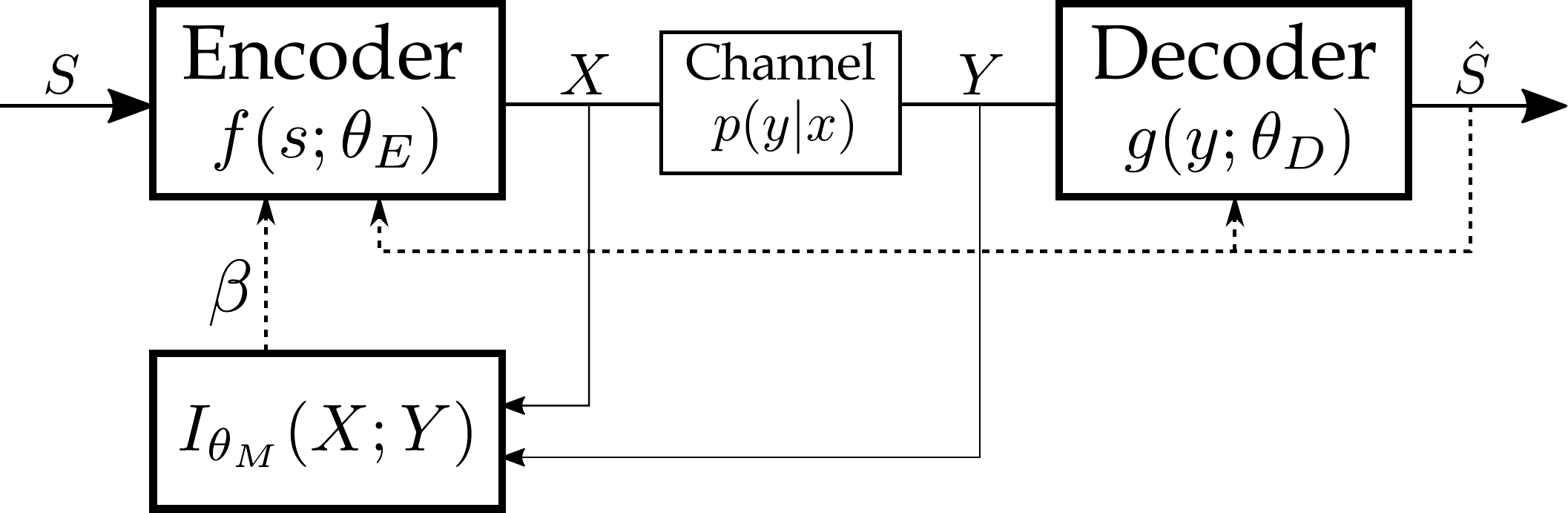}
	\caption{Rate-approaching autoencoder with the mutual information estimator block. The channel samples and the decoded message are used during the training process. The former allows to built an optimal encoding scheme exploiting the mutual information block, the latter, instead, allows to measure the decoding error through the cross-entropy loss function.}
	\label{fig:CE_MI}
\end{figure} 

The rationale behind the proposed method is formally discussed in the next section.

\section{Information Bottleneck Method and Cross-Entropy Decomposition}
\label{sec:IB}
Autoencoders are usually trained by minimizing the cross-entropy loss function \cite{OsheaIntro}. Nevertheless, cross-entropy can be minimized even for random labels as shown in \cite{Zhang2017}, leading to several overfitting issues. The work in \cite{Soatto2018} proved how a deep NN can just memorize the dataset (in its weights) to minimize the cross-entropy, yielding to poor generalization. Hence, the authors proposed an information bottleneck (IB) regularization term to prevent overfitting, similarly to the IB Lagrangian, originally presented in \cite{Tishby1999}. Indeed, the IB method optimally compresses the input random variable by eliminating the irrelevant features which do not contribute to the prediction of the output random variable. 

From an autoencoder-based communication systems point of view, let $S \rightarrow X \rightarrow Y$ be a prediction Markov chain, where $S$ represents the message to be sent, $X$ the compressed symbols and $Y$ the received symbols. The IB method solves
\begin{equation}
\label{eq:IB}
\mathcal{L}(p(x|s)) = I(S;X)-\beta I(X;Y),
\end{equation}
where the positive Lagrange multiplier $\beta$ plays as a trade-off parameter between the complexity of the encoding scheme (rate) and the amount of relevant information preserved in it. 

The communication chain adds an extra Markov chain constraint $Y\rightarrow \hat{S}$, where $\hat{S}$ represents the decoded message. Therefore, in order to deal with the full autoencoder chain, we can substitute the first term of the RHS in \eqref{eq:IB} with the cross-entropy loss function, as presented in \eqref{eq:CE_MI}. However, the Lagrange multiplier (or regularization parameter in ML terms) operates now as a trade-off parameter between the complexity to reconstruct the original message and the amount of information preserved in its compressed version.

To further motivate the choice for the new loss function with the mutual information as the regularization term, let us consider the following decomposition of the cross-entropy loss function.

\begin{lemma}\emph{(See \cite{Hoydis2019})}
\label{lemma:Lemma1}
Let $s \in \mathcal{M}$ be the transmitted message and let $(x,y)$ be samples drawn from the joint distribution $p_{XY}(x,y)$. If $x = f(s;\theta_E)$ is an invertible function representing the encoder and if $p_{\hat{S}|Y}(\hat{s}|y;\theta_D) = g(y;\theta_D) $ is the decoder block, then the cross-entropy function $\mathcal{L}(\theta_E,\theta_D)$ admits the following decomposition
\begin{align}
\mathcal{L}(\theta_E,\theta_D) &= H(S)-I_{\theta_E}(X;Y)+ \nonumber \\ 
&+\mathbb{E}_{y\sim p_Y(y)}[D_{\text{KL}}(p_{X|Y}(x|y)||p_{\hat{S}|Y}(x|y;\theta_D))].
\end{align}
\end{lemma}

The cross-entropy decomposition in Lemma \ref{lemma:Lemma1} can be read in the following way: the first two terms are responsible for the conditional entropy of the received symbols. In the particular case of a uniform source, only the mutual information between the transmitted and received symbols is controlled by the encoding function during the training process. On the contrary, the last term measures the error in computing the divergence between the true posterior distribution and the decoder-approximated one. As discussed before, the network could minimize the cross-entropy just by minimizing the KL-divergence, concentrating itself only on the label information (decoding) rather than on the symbol distribution (coding). The parameter $\beta$ in \eqref{eq:CE_MI} helps in balancing these two different contributions
\begin{align}
\label{eq:CE_MI_decomp1}
\mathcal{L}(\theta_E,\theta_M, \theta_D) &= H(S)-I_{\theta_E}(X;Y)-\beta I_{\theta_E, \theta_M}(X;Y)+ \nonumber \\ 
&+\mathbb{E}_{y\sim p_Y(y)}[D_{\text{KL}}(p_{X|Y}(x|y)||p_{\hat{S}|Y}(x|y;\theta_D))].
\end{align}
Moreover, if the mutual information estimator is consistent, \eqref{eq:CE_MI_decomp1} is equal to
\begin{align}
\label{eq:CE_MI_decomp2}
\mathcal{L}(\theta_E, \theta_D) &= H(S)-(\beta+1)I_{\theta_E}(X;Y)+ \nonumber \\ 
&+\mathbb{E}_{y\sim p_Y(y)}[D_{\text{KL}}(p_{X|Y}(x|y)||p_{\hat{S}|Y}(x|y;\theta_D))].
\end{align}
It is immediate to notice that for $\beta<-1$, the network gets in conflict since it would try to minimize both the mutual information and the KL-divergence. Therefore, optimal values for $\beta$ lie on the semi-line $\beta>-1$.

\section{Capacity-Approaching Autoencoder}
\label{sec:capacity}
\begin{algorithm}
\caption{Capacity Learning with Capacity-Approaching Autoencoders}
\label{alg:1}
\begin{algorithmic}[1]
\Inputs{$L$ SNR increasing values, $\epsilon$ threshold.}
\Initialize{$R_0 = k_0/n_0$ initial rate, $i=0, j = 0$.}
\For{$l=1$ to $L$}
	\State Train $AE^{(0)}(k_0,n_0)$;
	\State Compute $I^{(0)}_{\theta_M}(X;Y)$;
	\While{$\Delta>\epsilon$}
		\State $k_{i+1} = (R_{i}\cdot n_j) +1$;
		\State $R_{i+1}= k_{i+1}/n_j$;
		\State Train $AE^{(i+1)}(k_{i+1},n_j)$;
		\If{$R_{i+1}$ is not achievable}
			\State $n_{j+1}=n_j+1$;
			\State $j= j +1$;
		\Else
			\State Compute $I^{(i+1)}_{\theta_M}(X;Y)$;
			\State Evaluate $\Delta = \biggl|\frac{I^{(i+1)}_{\theta_M}(X;Y)-I^{(i)}_{\theta_M}(X;Y)}{I^{(i)}_{\theta_M}(X;Y)}\biggr|$;
		\EndIf
		\State $i=i+1$;
	\EndWhile
	\State $C_l = I^{(i)}_{\theta_M}(X;Y)$ estimated capacity.
\EndFor
\end{algorithmic}
\end{algorithm}
Interestingly, the mutual information block can be exploited to obtain an estimate of the channel capacity. The autoencoder-based system is subject to a power constraint coming from the transmitter hardware and it generally works at a fixed rate $R$ and channel uses $n$. For $R$ and $n$ fixed, the scheme discussed in Fig.\ref{fig:CE_MI} optimally designs the coded signal distribution and provides an estimate $I_{\theta_M}(X;Y)$ of the mutual information $I(X;Y)$ which approaches $R$. However, a question remains still open whether the achieved rate with the designed code is actually the channel capacity. 

To find the channel capacity $C$ and determine the optimal signal distribution $p_X(x)$, a broader search on the coding rate needs to be conducted, relaxing both the constraints on $R$ and $n$. The flexibility on $R$ and $n$ requires to use different autoencoders. In the following, we denote with $AE(k,n)$ a rate-approaching autoencoder-based system that transmits $n$ complex symbols at a rate $R=k/n$, where $k=\log_2(M)$ and $M$ is the number of possible messages. The proposed methodology can be segmented in two phases: 
\begin{enumerate}
\item Training of a rate-approaching autoencoder $AE(k,n)$ for a fixed coding rate $R$ and channel uses $n$, enabled via the loss function in \eqref{eq:CE_MI};
\item Adaptation of the coding rate $R$ to build capacity approaching codes $\mathbf{x}\sim p_X(x)$ and consequently find the channel capacity $C$.
\end{enumerate}
We remark that the capacity $C$ is the maximum data rate $R$ that can be transmitted through the channel at an arbitrarily small error probability. Therefore, the proposed algorithm makes an initial guess rate $R_0$ and slowly increases it by playing on both $k$ and $n$. 

The basic idea is to iteratively train a pair of autoencoders $AE^{(i)}(k_i,n_j), AE^{(i+1)}(k_{i+1},n_j)$ and evaluate both the mutual informations $I^{i}_{\theta_M}(X;Y), I^{(i+1)}_{\theta_M}(X;Y)$, at a fixed power constraint. The first autoencoder works at a rate $R_i=k_i/n_j$, while the second one at $R_{i+1}=k_{i+1}/n_j$, with $R_{i+1}>R_i$. If the ratio 
\begin{equation}
\biggl|\frac{I^{(i+1)}_{\theta_M}(X;Y)-I^{(i)}_{\theta_M}(X;Y)}{I^{(i)}_{\theta_M}(X;Y)}\biggr|<\epsilon,
\end{equation}
where $\epsilon$ is an input positive parameter, the code is reaching the capacity limit for the fixed power. If the rate $R_{i+1}$ is not achievable (it does not exist a error-free decoding scheme), a longer code $n_{j+1}$ is required. The algorithm in Tab.\ref{alg:1} describes the pseudocode that implements the channel capacity estimation and capacity-approaching code using as a building block the rate-approaching autoencoder.

\subsection{Remarks}
\label{subsec:remarks}
The proposed capacity-approaching autoencoder offers a constructive learning methodology to design a coding scheme that approaches capacity and to know what such a capacity is, even for channels that are unknown or for which a closed form expression for capacity does not exist. Indeed, training involves numerical procedures which may introduce some challenges. Firstly, the autoencoder is a NN and it is well known that its performance depends on the training procedure, architecture design and hyper-parameters tuning. Secondly, the MINE block converges to the true mutual information mostly for a low number of samples. In practice, when $n$ is large, the estimation often produces wrong results, therefore, a further investigation on stable numerical estimators via NNs needs to be conducted. Lastly, the autoencoder fails to scale with high code dimension. Indeed, for large values of $n$, the network could get stuck in local minima or, in the worst scenario, could not provide zero-error codes. The proposed approach transcends such limitations, although they have to be taken into account in the implementation phase.
On the positive side, the work follows the direction of explainable machine learning, in which the learning process is motivated by an information-theoretic approach.
Possible improvements are in defining an even tighter bound in \eqref{eq:lower_bound} and in adopting different network structures (convolutional or recurrent NNs). 

It should be noted that the approach works also for non linear channels where optimal codes have to be designed under an average power constraint and not for a given operating SNR which is appropriate for linear channels with additive noise.  

%% file: Results.tex
\section{Numerical Results}
\label{sec:results}
In this section, we present several results obtained with the rate-approaching autoencoders. Numerical results demonstrate an improvement in the decoding schemes (measured in terms of BLER) and show the achieved rates with respect to capacity in channels for which a closed form capacity formulation is known, such as the AWGN channel, and unknown, such as additive uniform noise and Rayleigh fading ones.

The following schemes consider an average power constraint at the transmitter side $\mathbb{E}[|\mathbf{x}|^2] = 1 $, implemented through a batch-normalization layer. Training of the end-to-end autoencoder is performed w.r.t. to the loss function in \eqref{eq:CE_MI}, implemented via a double minimization process since also the MINE block needs to be trained
\begin{equation}
\label{eq:CE_MI_final}
\min_{\theta_E,\theta_D} \min_{\theta_M} \mathbb{E}_{(x,y)\sim p_{XY}(x,y)}[-\log(g(y;\theta_D))]-\beta I_{\theta_M}(X;Y).
\end{equation}

Furthermore, the training procedure was conducted with the same number of iterations for different values of the regularization parameter $\beta$.
We used Keras with TensorFlow \cite{TensorFlow} as backend to implement the proposed rate-approaching autoencoder. The code has been tested on a Windows-based operating system provided with Python 3.6, TensorFlow 1.13.1, Intel core i7-3820 CPU. The implementation is not straightforward, thus, to allow reproducible results and for clarity, the code will be rendered publicly available. \footnote{\url{https://github.com/tonellolab/capacity-approaching-autoencoders}}

\subsection{Coding-decoding capability}
\begin{figure}
	\centering
	\includegraphics[scale=0.34]{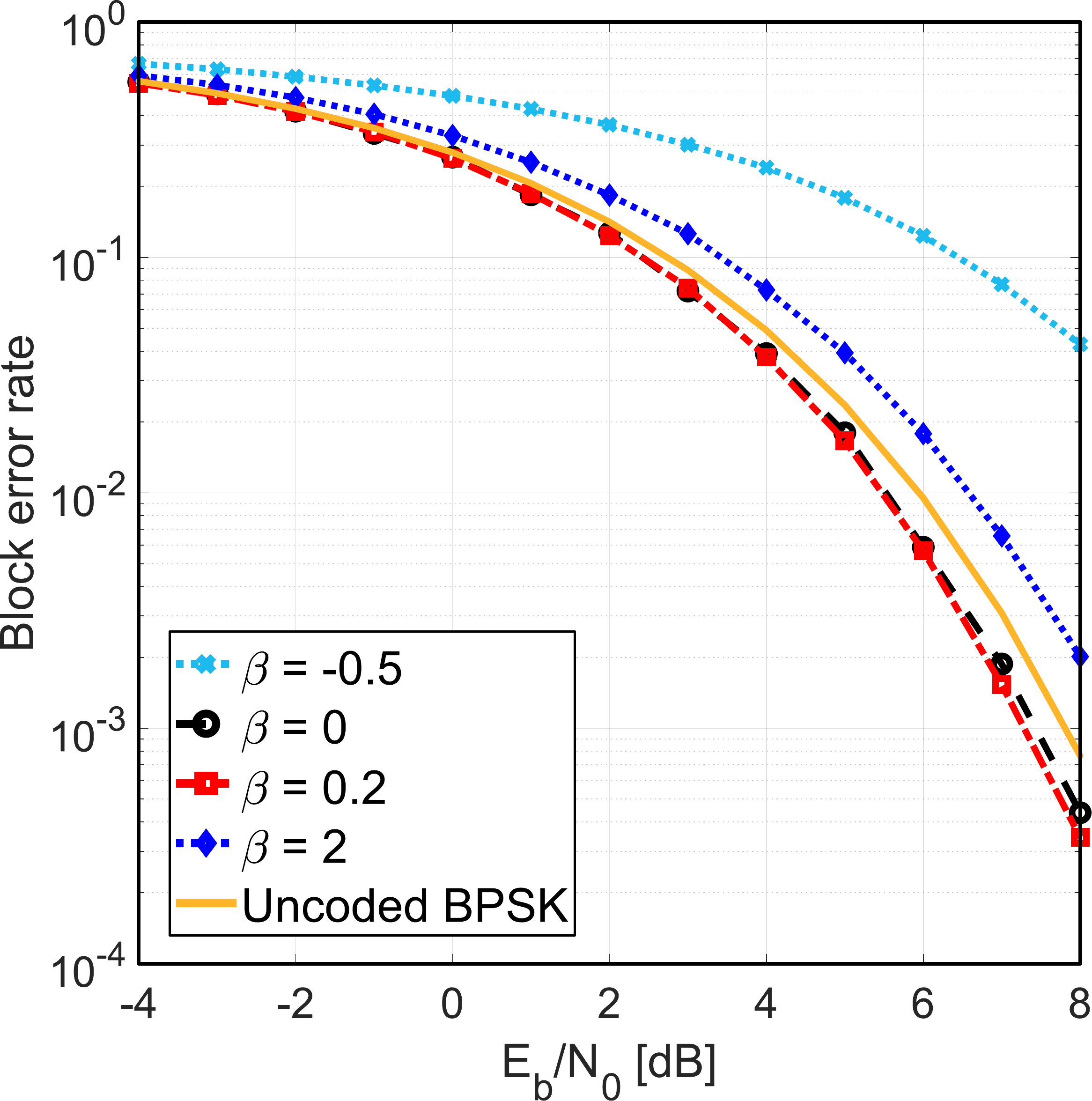}
	\caption{BLER of the rate-approaching autoencoder AE($4,2$) for different values of the regularization parameter $\beta$ for an average power constraint, $k=4$ and $n=2$.}
	\label{fig:Ber_4_4}
\end{figure}

As first experiment, we consider a rate-approaching autoencoder AE($4,2$) with rate $R=2$. To demonstrate the effective influence on the performance of the mutual information term controlled by $\beta$ in \eqref{eq:CE_MI_final}, we investigate $4$ different representative values of the regularization parameter. Fig.\ref{fig:Ber_4_4} illustrates the obtained BLER after the same number of training iterations. We notice that the lowest BLER is achieved for $\beta=0.2$, therefore as expected, the mutual information contributes in finding better encoding schemes. Despite the small gain, the result highlights that better BLER can be obtained using the same number of iterations. As shown in \eqref{eq:CE_MI_decomp2}, negative values of $\beta$ tend to force the network to just memorize the dataset, while large positive values create an unbalance. We remark that $\beta=0$ corresponds to the classic autoencoder approach proposed in \cite{OsheaIntro}. To study optimal values of $\beta$, a possible approach could try to find the value of $\beta$ for which the two gradients (cross-entropy and mutual information) are equal in magnitude. In the following, we assume $\beta=0.2$.

\begin{figure}
	\centering
	\includegraphics[scale=0.35]{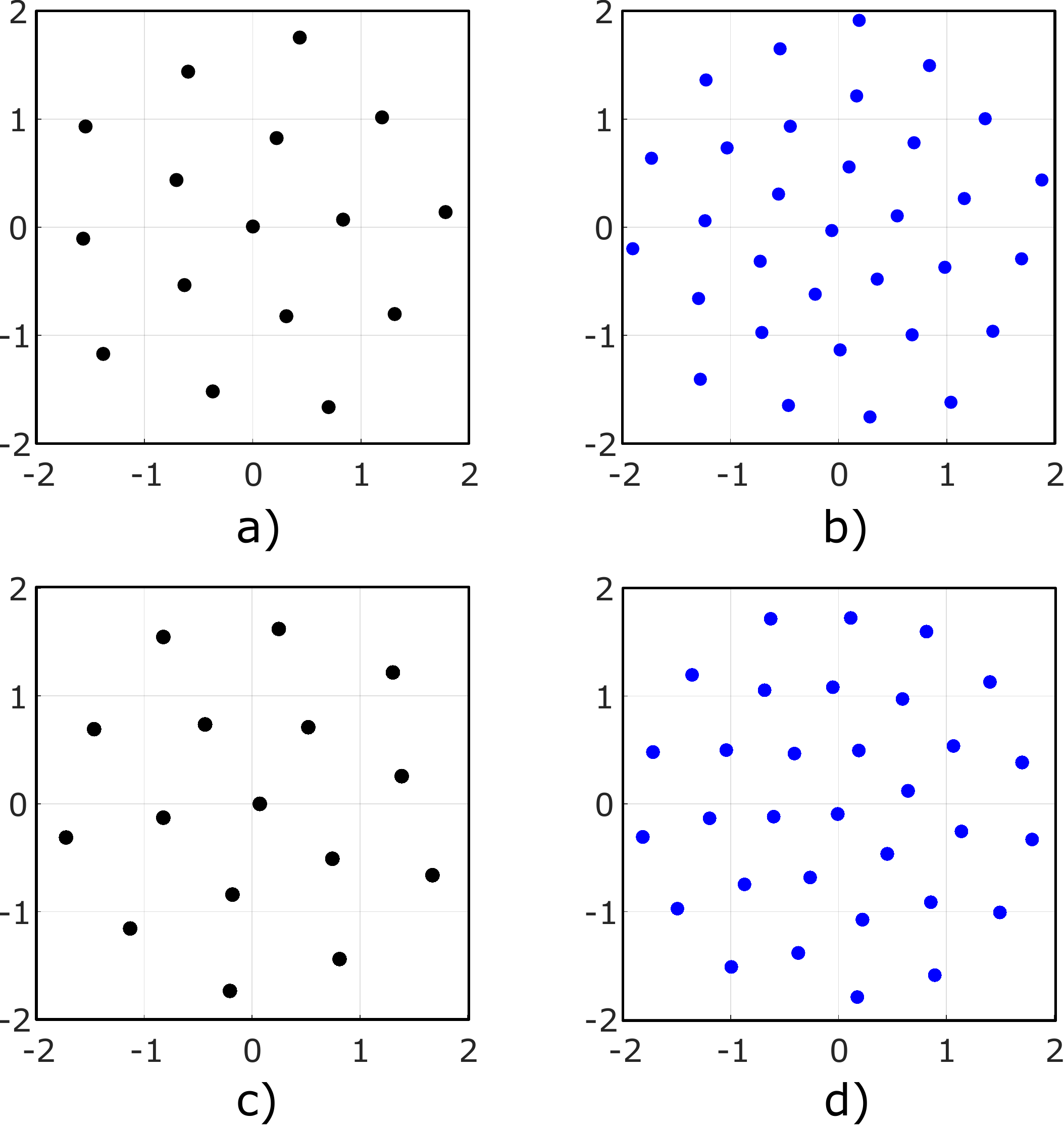}
	\caption{Constellation designed by the encoder during the end-to-end training with $\beta=0.2$ and parameters $(k,n)$: a) ($4,1$), b) ($5,1$), c) 2 dimensional t-SNE of AE($4,2$) and d)  2 dimensional t-SNE of AE($5,2$).}
	\label{fig:constellation}
\end{figure}

To assess the methodology even with higher dimension $M$ of the input alphabet, we illustrate the optimal constellation schemes when the number of possible messages is $M=16$ and $M=32$. Moreover, two cases are studied, when we transmit one complex symbol ($n=1$) and two dependent complex symbols ($n=2$) over the channel. Fig.\ref{fig:constellation}a and Fig.\ref{fig:constellation}b show the learned pentagonal and hexagonal grid constellations when only one symbol is transmitted for an alphabet dimension of $M=16$ and $M=32$, respectively. Fig.\ref{fig:constellation}c and Fig.\ref{fig:constellation}d show, instead, an optimal projection of the coded signals in a 2D space through the learned two-dimensional t-distributed stochastic neighbor embedding
(t-SNE) \cite{vandermaaten2008visualizing}. We notice that the two pairs of constellations are similar, and therefore, even for larger code-lengths, the mutual information pushes the autoencoder to learn the optimal signal constellation.

\subsection{Capacity-approaching codes over different channels}
The mutual information block inside the autoencoder can be exploited to design capacity-approaching codes, as discussed in Sec.\ref{sec:capacity}. To show the potentiality of the method, we analyze the achieved rate, e.g. the mutual information, in three different scenarios. The first one considers the transmission over an AWGN channel, for which we know the exact closed form capacity. The second and third ones, instead, consider the transmission over an additive uniform noise and over a Rayleigh fading channel, for which we do not know the capacity in closed form. However, we expect the estimated mutual information to be a tight lower bound for the real channel capacity, especially at low SNRs.

\subsubsection{AWGN channel}
It is well known that for a discrete memory-less channel with input-output relation given by (assuming complex signals)
\begin{equation}
\label{eq:additive_noise}
Y_i=X_i+N_i,
\end{equation}
where the noise samples $N_i\sim \mathcal{CN}(0,\sigma^2)$ are i.i.d. and independent of $X_i$, and with a power constraint on the input signal $\mathbb{E}[|X_i|^2] \leq P$, the channel capacity is achieved by $X_i\sim \mathcal{CN}(0,P)$ and is equal to
\begin{equation}
C = \log_2(1+\text{SNR}) \; \; \text{[bits/ channel use]}.
\end{equation}

\begin{figure}
	\centering
	\includegraphics[scale=0.34]{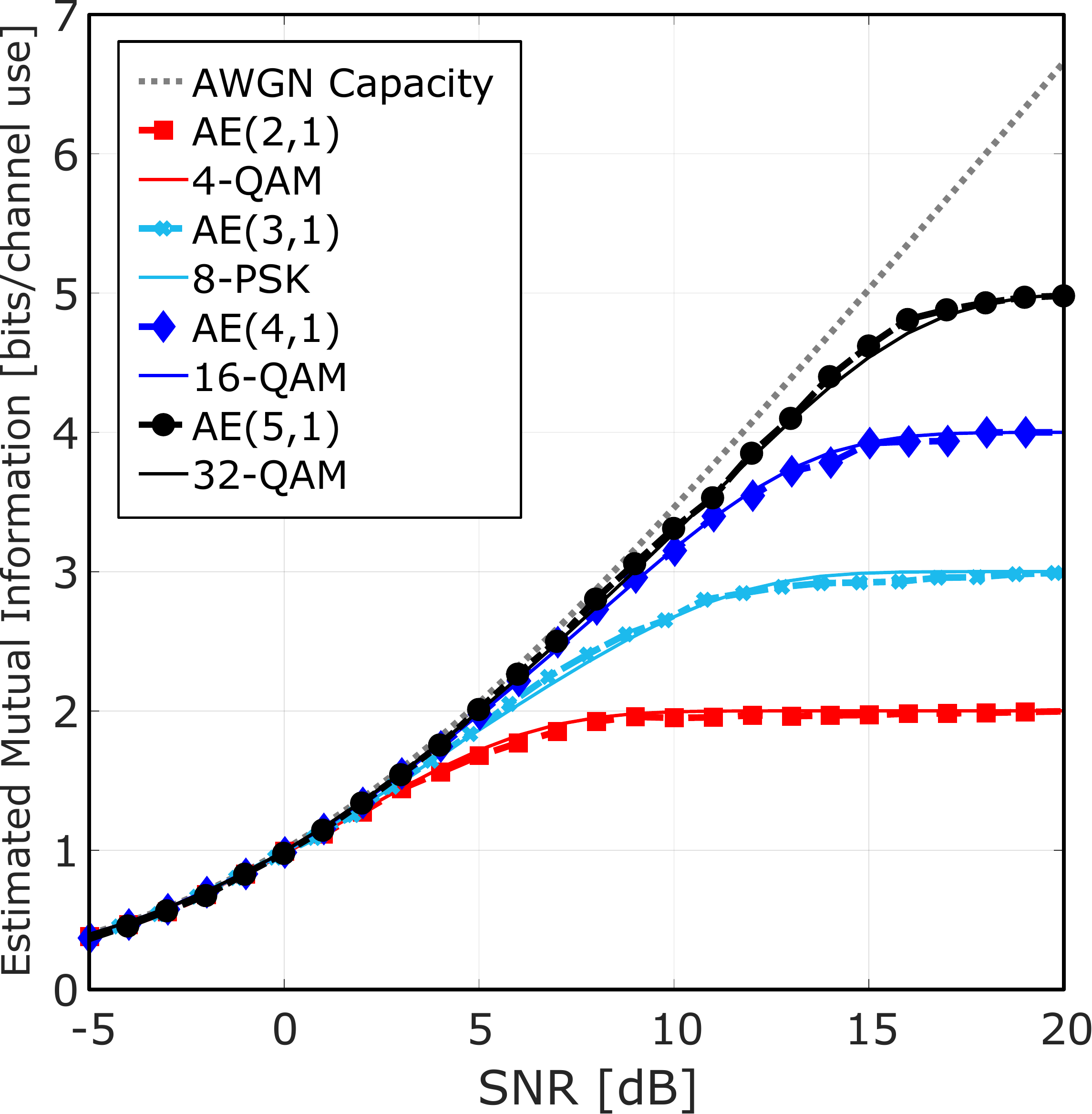}
	\caption{Estimated mutual information achieved with $\beta=0.2$ for different dimension of the alphabet $M$ when only one symbol ($n=1$) is transmitted over an AWGN channel.}
	\label{fig:gaussian_capacity}
\end{figure}
The rate-approaching autoencoder attempts to maximize the mutual information during the training progress by modifying, at each iteration, the input distribution $p_X(x)$. Thus, given the input parameters, it produces optimal codes for which the estimation of the mutual information is already provided by MINE. Fig.\ref{fig:gaussian_capacity} illustrates the estimated mutual information when $\beta=0.2$ for different values of the alphabet cardinality $M$. A comparison is finally made with established $M$-QAM schemes. We remark that for discrete-input signals with distribution $p_X(x)$, the mutual information is given by
\begin{equation}
I(X;Y) = \sum_x p_X(x)\cdot \mathbb{E}_{p_Y(y|x)}\biggl[\log \frac{p_Y(y|x)}{p_Y(y)}\biggr],
\end{equation}

and in particular with uniformly distributed symbols (only
geometric shaping), $p_X(x)=1/M$. It is found that the autoencoder constructively provides a good estimate of mutual information. In addition, for the case $M=32$, the conventional QAM signal constellation is not optimal, since the autoencoder AE($5,1$) performs geometric signal shaping and finds a constellation that can offer higher information rate as it is visible in the range $13-16$ dB.  

\subsubsection{Additive uniform noise channel}
\begin{figure}
	\centering
	\includegraphics[scale=0.34]{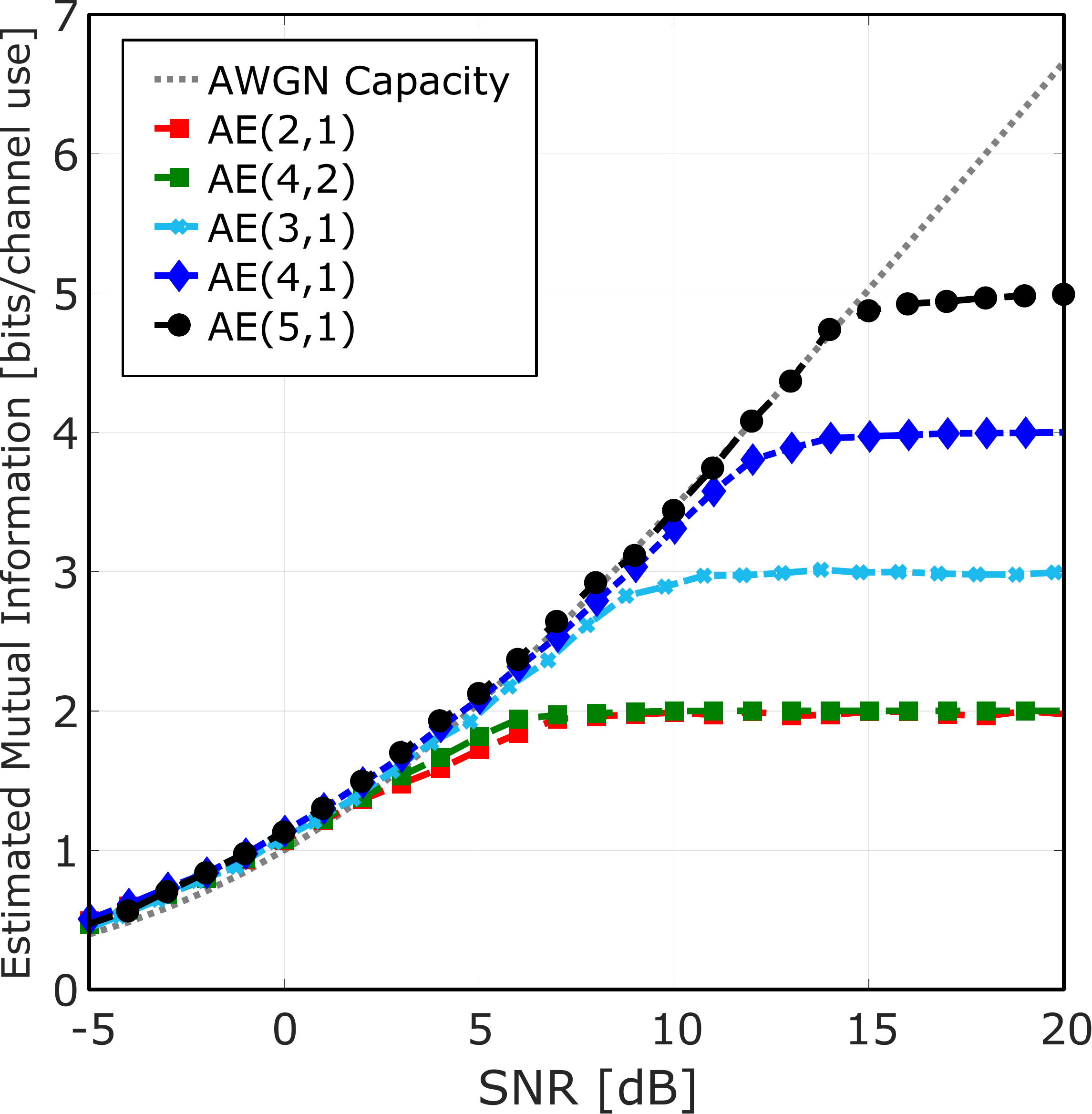}
	\caption{Estimated mutual information achieved with $\beta=0.2$ for different dimension of the alphabet $M$ when one and two symbols ($n=1,2$) are transmitted over an additive uniform noise channel.}
	\label{fig:additive_capacity}
\end{figure}

\begin{figure}
	\centering
	\includegraphics[scale=0.20]{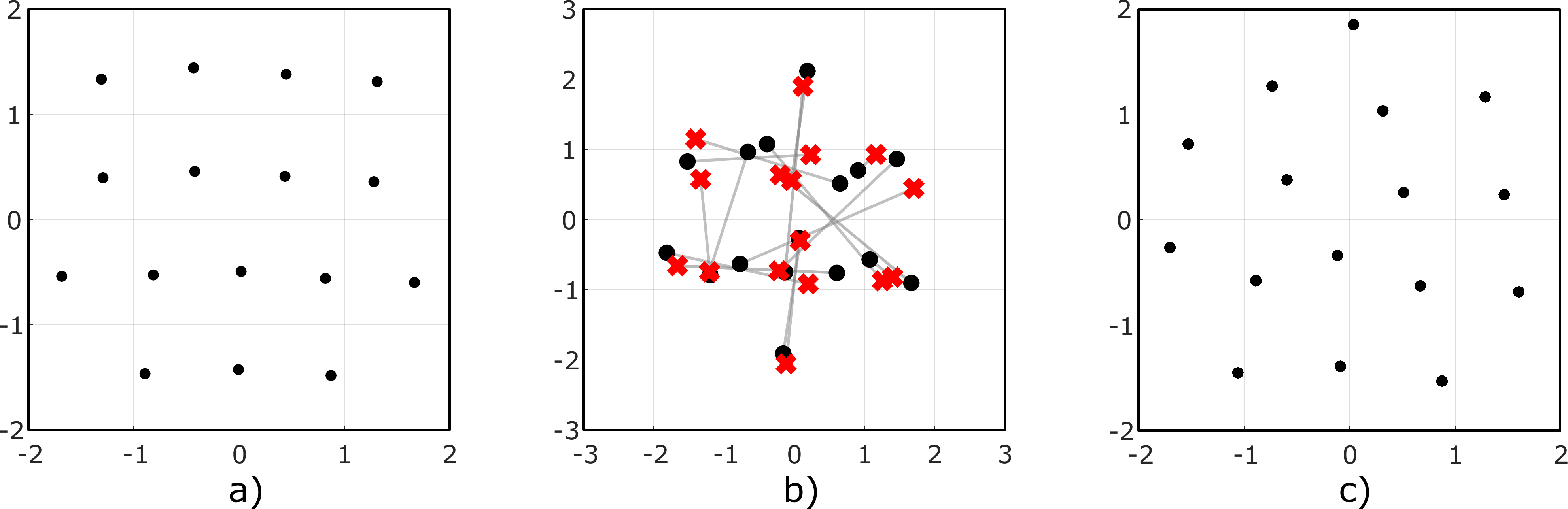}
	\caption{Constellation designed by the encoder during the end-to-end training with $\beta=0.2$, uniform noise layer and parameters $(k,n)$: a) ($4,1$), b) pairs of transmitted coded symbols AE($4,2$), c) 2 dimensional t-SNE of AE($4,2$).}
	\label{fig:constellation_unif}
\end{figure}

No closed form capacity expression is known when the noise $N$ has uniform density $N\sim \mathcal{U}(-\frac{\Delta}{2},\frac{\Delta}{2})$ under an average power constraint. However, Shannon proved that the AWGN capacity is the lowest among all additive noise channels of the form \eqref{eq:additive_noise}. Consistently, it is rather interesting to notice that the estimated mutual information for the uniform channel is higher than the AWGN capacity for low SNRs until it saturates to the coding rate $R$, as depicted in Fig.\ref{fig:additive_capacity}. 
Moreover, differently from the AWGN coding signal set, here we also consider the complex signals generated by the encoder over $2$ channel uses. As expected, the mutual information achieved by the code produced with the autoencoder AE($4,2$) is higher than the AE($2,1$) one, consistently with the idea that $n>1$ introduces a temporal dependence in the code to improve the decoding phase. In addition, Fig.\ref{fig:constellation_unif}a illustrates the constellation produced by the rate-approaching autoencoders AE($4,1$) in the uniform noise case, Fig.\ref{fig:constellation_unif}b shows how the transmitted coded symbols (transmitted complex coefficients) vary for different channel uses while Fig.\ref{fig:constellation_unif}c, instead, displays the learned two-dimensional t-SNE constellation of the code produced by the AE($4,2$). 

\subsubsection{Rayleigh channel}
\begin{figure}
	\centering
	\includegraphics[scale=0.34]{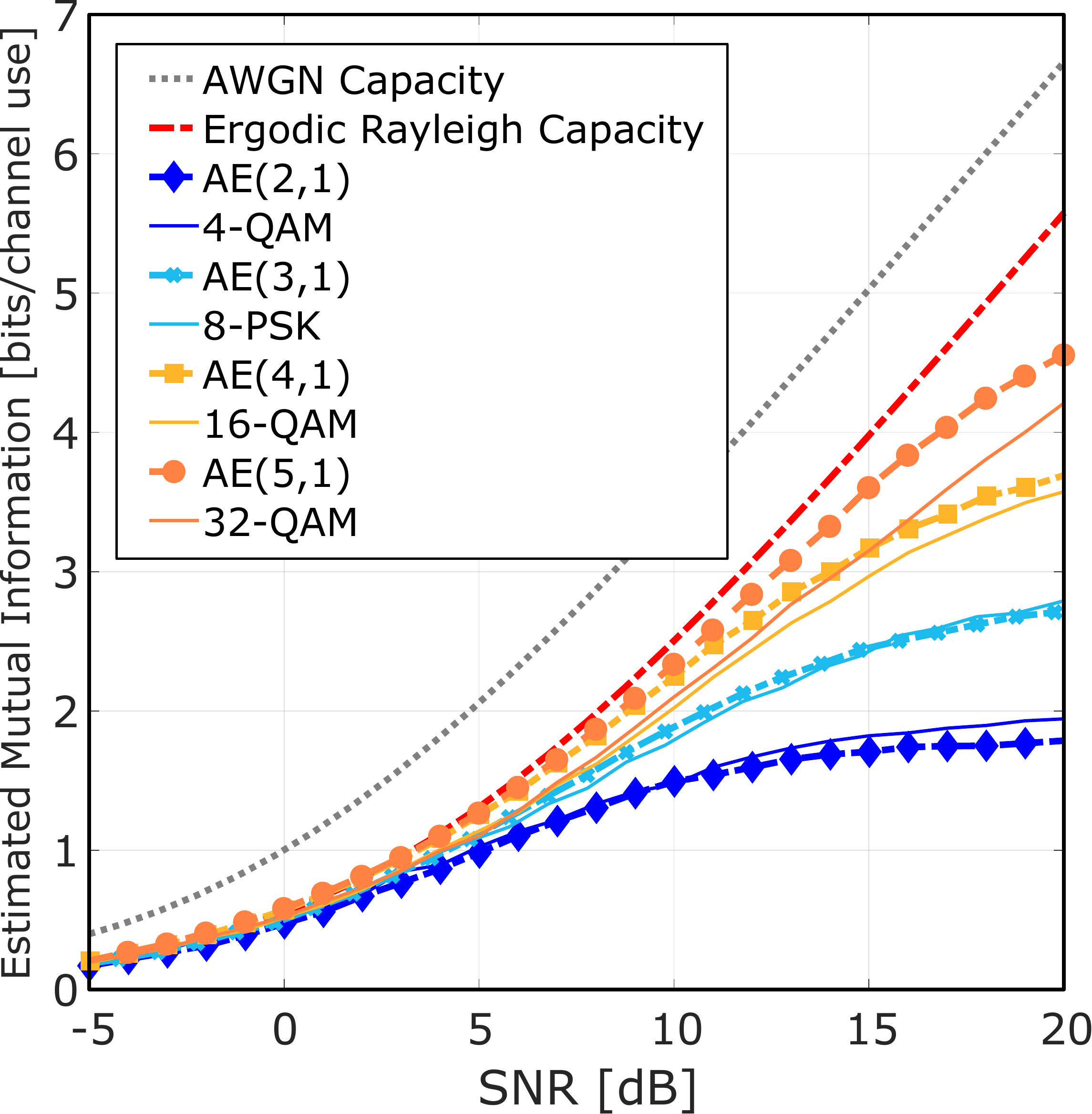}
	\caption{Estimated mutual information achieved with $\beta=0.2$ for different dimension of the alphabet $M$ when one symbol ($n=1$) is transmitted over a Rayleigh channel.}
	\label{fig:rayleigh_capacity}
\end{figure}
As final experiment, we introduce fading in the communication channel, in particular we consider a Rayleigh fading channel of the form
\begin{equation}
\label{eq:rayleigh_noise}
Y_i=h_iX_i+N_i,
\end{equation}
where $N_i\sim \mathcal{CN}(0,\sigma^2)$ and $h_i$ is a realization of a random variable whose amplitude $\alpha$ belongs to the Rayleigh distribution $p_R(r)$ and is independent of the signal and noise. The ergodic capacity is given by
\begin{equation}
C = \mathbb{E}_{\alpha\sim p_R(r)}\biggl[\log_2(1+\alpha ^2 \cdot \text{SNR})\biggr].
\end{equation}
Fig.\ref{fig:rayleigh_capacity} shows the estimated mutual information between the transmitted and received signals over a Rayleigh channel with several alphabet dimensions $M$ and compares it with the conventional $M$-QAM schemes. In all the cases, the estimated  mutual information is upper bounded by the ergodic Rayleigh capacity. Similarly to the uniform case, it is curious to notice that the achieved information rate with the rate-approaching autoencoder is in some cases higher than the one obtained with the $M$-QAM schemes.
In particular, with $M=32$ the AE($5,1$) exceeds by $0.5$ bit/channel uses at SNR$=15$ dB the $32$-QAM scheme. 

%% file: Conclusions.tex
\section{Conclusions}
\label{sec:conclusions}
This paper has firstly discussed the autoencoder-based communication system, highlighting the limits of the current cross-entropy loss function used for the training process. A regularization term that accounts for the mutual information between the transmitted and the received signals has been introduced to design optimal coding-decoding schemes for fixed rate $R$, code-length $n$ and given power constraint. The rationale behind the mutual information choice has been motivated exploiting the information bottleneck principle and the fundamental concept of channel capacity. In addition, an adaptation of the coding rate $R$ allowed us to build a capacity learning algorithm enabled by the novel loss function in a scheme named capacity-approaching autoencoder. Remarkably, the presented methodology does not make use of any theoretical a-priori knowledge of the communication channel and therefore opens the door to several future studies on intractable channel models, an example of which is the power line communication channel.

%% file: Appendix.tex
\appendix
\section{Appendix}
 \setcounter{lemma}{0}

The proof of Lemma 1 follows some steps provided in \cite{Hoydis2019}. 
\begin{lemma}\emph{(See \cite{Hoydis2019})}
\label{lemma:Lemma1}
Let $s \in \mathcal{M}$ be the transmitted message and let $(x,y)$ be samples drawn from the joint distribution $p_{XY}(x,y)$. If $x = f(s;\theta_E)$ is an invertible function representing the encoder and if $p_{\hat{S}|Y}(\hat{s}|y;\theta_D) = g(y;\theta_D) $ is the decoder block, then the cross-entropy function $\mathcal{L}(\theta_E,\theta_D)$ admits the following decomposition
\begin{align}
\mathcal{L}(\theta_E,\theta_D) &= H(S)-I_{\theta_E}(X;Y)+ \nonumber \\ 
&+\mathbb{E}_{y\sim p_Y(y)}[D_{\text{KL}}(p_{X|Y}(x|y)||p_{\hat{S}|Y}(x|y;\theta_D))]. \nonumber
\end{align}
\end{lemma}
\begin{proof}
The cross-entropy loss function can be rewritten as follows
\begin{align*}
& \mathcal{L}(\theta_E,\theta_D) = \mathbb{E}_{(x,y)\sim p_{XY}(x,y)}[-\log(p_{\hat{S}|Y}(\hat{s}|y;\theta_D))] \\ \nonumber
&= -\sum_{x,y}{p_{XY}(x,y)\log(p_{\hat{S}|Y}(\hat{s}|y;\theta_D))} \\ \nonumber
&= -\sum_{x,y}{p_{XY}(x,y)\log(p_X(x))} + \\
&+ \sum_{x,y}{p_{XY}(x,y)\log\biggl(\frac{p_X(x)}{p_{\hat{S}|Y}(\hat{s}|y;\theta_D)}\biggr)} \nonumber
\end{align*}
Using the encoder hypothesis, the first term in the last expression corresponds to the source entropy $H(S)$. Therefore,
\begin{align*}
&\mathcal{L}(\theta_E,\theta_D) = H(S) + \sum_{x,y}{p_{XY}(x,y)\log\biggl(\frac{p_X(x)\cdot p_Y(y)}{p_{XY}(x,y)}\biggr)}+\\ \nonumber
&+ \sum_{x,y}{p_{XY}(x,y)\log\biggl(\frac{p_{XY}(x,y)}{p_{\hat{S}|Y}(\hat{s}|y;\theta_D)\cdot p_Y(y)}\biggr)} \\ \nonumber
&= H(S)-I(X;Y)+\sum_{x,y}{p_{XY}(x,y)\log\biggl(\frac{p_{X|Y}(x|y)}{p_{\hat{S}|Y}(\hat{s}|y;\theta_D)}\biggr)} \\ \nonumber
&= H(S)-I_{\theta_E}(X;Y)+\mathbb{E}_{y}[D_{\text{KL}}(p_{X|Y}(x|y)|| p_{\hat{S}|Y}(x|y;\theta_D))] \qedhere \nonumber
\end{align*} 
\end{proof}

%% file: MAIN.bbl
 \newcommand{\noop}[1]{}
\begin{thebibliography}{10}
\providecommand{\url}[1]{#1}
\csname url@samestyle\endcsname
\providecommand{\newblock}{\relax}
\providecommand{\bibinfo}[2]{#2}
\providecommand{\BIBentrySTDinterwordspacing}{\spaceskip=0pt\relax}
\providecommand{\BIBentryALTinterwordstretchfactor}{4}
\providecommand{\BIBentryALTinterwordspacing}{\spaceskip=\fontdimen2\font plus
\BIBentryALTinterwordstretchfactor\fontdimen3\font minus
  \fontdimen4\font\relax}
\providecommand{\BIBforeignlanguage}[2]{{%
\expandafter\ifx\csname l@#1\endcsname\relax
\typeout{** WARNING: IEEEtran.bst: No hyphenation pattern has been}%
\typeout{** loaded for the language `#1'. Using the pattern for}%
\typeout{** the default language instead.}%
\else
\language=\csname l@#1\endcsname
\fi
#2}}
\providecommand{\BIBdecl}{\relax}
\BIBdecl

\bibitem{Shannon1948}
\BIBentryALTinterwordspacing
C.~E. Shannon, ``A mathematical theory of communication,'' \emph{The Bell
  System Technical Journal}, vol.~27, no.~3, pp. 379--423, 7 1948. [Online].
  Available: \url{https://ieeexplore.ieee.org/document/6773024/}
\BIBentrySTDinterwordspacing

\bibitem{OsheaIntro}
T.~{O'Shea} and J.~{Hoydis}, ``An introduction to deep learning for the
  physical layer,'' \emph{IEEE Transactions on Cognitive Communications and
  Networking}, vol.~3, no.~4, pp. 563--575, Dec 2017.

\bibitem{ML_PLC}
A.~M. {Tonello}, N.~A. {Letizia}, D.~{Righini}, and F.~{Marcuzzi}, ``Machine
  learning tips and tricks for power line communications,'' \emph{IEEE Access},
  vol.~7, pp. 82\,434--82\,452, 2019.

\bibitem{Dorner2018}
S.~{Dorner}, S.~{Cammerer}, J.~{Hoydis}, and S.~t.~{Brink}, ``Deep learning
  based communication over the air,'' \emph{IEEE Journal of Selected Topics in
  Signal Processing}, vol.~12, no.~1, pp. 132--143, Feb 2018.

\bibitem{TurboAE}
Y.~Jiang, H.~Kim, H.~Asnani, S.~Kannan, S.~Oh, and P.~Viswanath, ``Turbo
  autoencoder: Deep learning based channel codes for point-to-point
  communication channels,'' in \emph{Advances in Neural Information Processing
  Systems 32}.\hskip 1em plus 0.5em minus 0.4em\relax Curran Associates, Inc.,
  2019, pp. 2758--2768.

\bibitem{Hoydis2019}
M.~{Stark}, F.~{Ait Aoudia}, and J.~{Hoydis}, ``Joint learning of geometric and
  probabilistic constellation shaping,'' in \emph{2019 IEEE Globecom Workshops
  (GC Wkshps)}, 2019, pp. 1--6.

\bibitem{Wunder2019}
R.~{Fritschek}, R.~F. {Schaefer}, and G.~{Wunder}, ``Deep learning for channel
  coding via neural mutual information estimation,'' in \emph{2019 IEEE 20th
  International Workshop on Signal Processing Advances in Wireless
  Communications (SPAWC)}, 2019, pp. 1--5.

\bibitem{Goodfellow2014}
I.~J. Goodfellow, J.~Pouget-Abadie, M.~Mirza, B.~Xu, D.~Warde-Farley, S.~Ozair,
  A.~Courville, and Y.~Bengio, ``Generative adversarial nets,'' in
  \emph{Proceedings of the 27th International Conference on Neural Information
  Processing Systems - Volume 2}, ser. NIPS'14.\hskip 1em plus 0.5em minus
  0.4em\relax Cambridge, MA, USA: MIT Press, 2014, pp. 2672--2680.

\bibitem{OsheaGAN}
T.~J. O'Shea, T.~Roy, N.~West, and B.~C. Hilburn, ``Physical layer
  communications system design over-the-air using adversarial networks,''
  \emph{2018 26th European Signal Processing Conference (EUSIPCO)}, pp.
  529--532, 2018.

\bibitem{RighiniLetizia2019}
D.~{Righini}, N.~A. {Letizia}, and A.~M. {Tonello}, ``Synthetic power line
  communications channel generation with autoencoders and gans,'' in \emph{2019
  IEEE International Conference on Communications, Control, and Computing
  Technologies for Smart Grids (SmartGridComm)}, 2019, pp. 1--6.

\bibitem{Letizia2019a}
N.~A. {Letizia}, A.~M. {Tonello}, and D.~{Righini}, ``Learning to synthesize
  noise: The multiple conductor power line case,'' in \emph{2020 IEEE
  International Symposium on Power Line Communications and its Applications
  (ISPLC)}, May 2020, pp. 1--6.

\bibitem{Kramer1991}
M.~A. Kramer, ``Nonlinear principal component analysis using autoassociative
  neural networks,'' \emph{AIChE Journal}, vol.~37, no.~2, pp. 233--243, 1991.

\bibitem{Vincent2008}
P.~Vincent, H.~Larochelle, Y.~Bengio, and P.-A. Manzagol, ``Extracting and
  composing robust features with denoising autoencoders,'' in \emph{Proceedings
  of the 25th International Conference on Machine Learning}, ser. ICML
  '08.\hskip 1em plus 0.5em minus 0.4em\relax New York, NY, USA: ACM, 2008, pp.
  1096--1103.

\bibitem{Rifai2011}
S.~Rifai, P.~Vincent, X.~Muller, X.~Glorot, and Y.~Bengio, ``Contracting
  auto-encoders: Explicit invariance during feature extraction,'' in
  \emph{Proceedings of the Twenty-eight International Conference on Machine
  Learning (ICML'11}, 2011.

\bibitem{Makhzani2013}
A.~Makhzani and B.~J. Frey, ``k-sparse autoencoders,'' in \emph{2nd
  International Conference on Learning Representations, {ICLR} 2014, Banff, AB,
  Canada, April 14-16, 2014, Conference Track Proceedings}, 2014.

\bibitem{Kingma2013}
D.~P. Kingma and M.~Welling, ``Auto-encoding variational bayes,'' in \emph{2nd
  International Conference on Learning Representations, {ICLR} 2014, Banff, AB,
  Canada, April 14-16, 2014, Conference Track Proceedings}, 2014.

\bibitem{Tishby1999}
N.~Tishby, F.~C. Pereira, and W.~Bialek, ``The information bottleneck method,''
  in \emph{Proc. of the 37-th Annual Allerton Conference on Communication,
  Control and Computing}, 1999, pp. 368--377.

\bibitem{Alemi2017}
A.~A. Alemi, I.~Fischer, J.~V. Dillon, and K.~Murphy, ``Deep variational
  information bottleneck,'' in \emph{5th International Conference on Learning
  Representations, {ICLR} 2017, Toulon, France, April 24-26, 2017, Conference
  Track Proceedings}, 2017.

\bibitem{Soatto2018}
A.~Achille and S.~Soatto, ``Emergence of invariance and disentanglement in deep
  representations,'' \emph{J. Mach. Learn. Res.}, vol.~19, pp. 50:1--50:34,
  2018.

\bibitem{Mine2018}
M.~I. Belghazi, A.~Baratin, S.~Rajeshwar, S.~Ozair, Y.~Bengio, A.~Courville,
  and D.~Hjelm, ``Mutual information neural estimation,'' in \emph{Proceedings
  of the 35th International Conference on Machine Learning}, ser. Proceedings
  of Machine Learning Research, vol.~80.\hskip 1em plus 0.5em minus 0.4em\relax
  Stockholmsmässan, Stockholm Sweden: PMLR, 10--15 Jul 2018, pp. 531--540.

\bibitem{OFDM_AE}
A.~{Felix}, S.~{Cammerer}, S.~{Dörner}, J.~{Hoydis}, and S.~{Ten Brink},
  ``Ofdm-autoencoder for end-to-end learning of communications systems,'' in
  \emph{2018 IEEE 19th International Workshop on Signal Processing Advances in
  Wireless Communications (SPAWC)}, 2018, pp. 1--5.

\bibitem{Tishby2017}
R.~Shwartz{-}Ziv and N.~Tishby, ``Opening the black box of deep neural networks
  via information,'' \emph{CoRR}, vol. abs/1703.00810, 2017.

\bibitem{279181}
Y.~{Bengio}, P.~{Simard}, and P.~{Frasconi}, ``Learning long-term dependencies
  with gradient descent is difficult,'' \emph{IEEE Transactions on Neural
  Networks}, vol.~5, no.~2, pp. 157--166, March 1994.

\bibitem{Moon1995}
Y.-I. Moon, B.~Rajagopalan, and U.~Lall, ``Estimation of mutual information
  using kernel density estimators,'' \emph{Phys. Rev. E}, vol.~52, pp.
  2318--2321, Sep 1995.

\bibitem{Kraskov2004}
A.~Kraskov, H.~St\"ogbauer, and P.~Grassberger, ``Estimating mutual
  information,'' \emph{Phys. Rev. E}, vol.~69, p. 066138, Jun 2004.

\bibitem{Nguyen2010}
X.~{Nguyen}, M.~J. {Wainwright}, and M.~I. {Jordan}, ``Estimating divergence
  functionals and the likelihood ratio by convex risk minimization,''
  \emph{IEEE Transactions on Information Theory}, vol.~56, no.~11, pp.
  5847--5861, 2010.

\bibitem{Donsker1983}
M.~D. Donsker and S.~R.~S. Varadhan, ``Asymptotic evaluation of certain markov
  process expectations for large time. iv,'' \emph{Communications on Pure and
  Applied Mathematics}, vol.~36, no.~2, pp. 183--212, 1983.

\bibitem{Zhang2017}
C.~Zhang, S.~Bengio, M.~Hardt, B.~Recht, and O.~Vinyals, ``Understanding deep
  learning requires rethinking generalization,'' in \emph{5th International
  Conference on Learning Representations, {ICLR} 2017, Toulon, France, April
  24-26, 2017, Conference Track Proceedings}, 2017.

\bibitem{TensorFlow}
M.~Abadi \emph{et~al.}, ``Tensorflow: Large-scale machine learning on
  heterogeneous distributed systems,'' \emph{CoRR}, vol. abs/1603.04467, 2016.

\bibitem{vandermaaten2008visualizing}
L.~van~der Maaten and G.~Hinton, ``Visualizing data using {t-SNE},''
  \emph{Journal of Machine Learning Research}, vol.~9, pp. 2579--2605, 2008.

\end{thebibliography}
